\documentclass[a4paper]{amsart}

\usepackage{amsmath}
\usepackage{amsfonts}
\usepackage{amssymb}
\usepackage{algorithm}
\usepackage{algorithmic}

\newcommand{\uinvnorm}{|\kern-1pt|\kern-1pt|}

\newcommand{\poly}{\operatorname{poly}}

\newcommand{\ket}[1]{\vert #1 \rangle}
\newcommand{\bra}[1]{\langle #1 \vert}

\theoremstyle{plain}
\newtheorem{theorem}{Theorem}

\newtheorem{proposition}[theorem]{Proposition}

\theoremstyle{definition}

\theoremstyle{remark}

\begin{document}

\bibliographystyle{amsplain}

\title{A quantum algorithm to solve nonlinear differential equations}
\author{Sarah K.\ Leyton}
\author{Tobias J.\ Osborne}
\thanks{We'd like to thank Aram Harrow for helpful correspondence. SKL was supported by the EPSRC and TJO was supported, in part, by the University of London central research fund.}
\address{Department of Mathematics, Royal Holloway, University of London, Egham, TW20 0EX, UK}
\email{tobias.osborne@rhul.ac.uk}

\begin{abstract}
In this paper we describe a quantum algorithm to solve sparse
systems of nonlinear differential equations whose nonlinear terms
are polynomials. The algorithm is nondeterministic and its expected
resource requirements are polylogarithmic in the number of variables
and exponential in the integration time. The best classical
algorithm runs in a time scaling linearly with the number of
variables, so this provides an exponential improvement. The
algorithm is built on two subroutines: (i) a quantum algorithm to
implement a nonlinear transformation of the probability amplitudes
of an unknown quantum state; and (ii) a quantum implementation of
Euler's method.
\end{abstract}

\maketitle

\section{Introduction}
Systems of nonlinear differential equations arise in an astounding
number of applications across the sciences ranging from engineering,
biological systems, and mathematics. Their theory occupies a large
proportion of the scientific literature and many subtle and profound
techniques have been developed to solve them. The numerical solution
of nonlinear ODEs is now a mature and well-established topic (see,
eg., \cite{press:2007a}) and there are many stable and efficient
classical algorithms to numerically integrate these equations,
ranging from the ``workhorse'' Runge-Kutta method, to
predictor-corrector methods and implicit methods.

The discovery of quantum algorithms (see, eg., \cite{nielsen:2000a}
for an introduction to quantum computation) has ushered in a new era
where previously intractable problems can now be theoretically
solved efficiently on a quantum computer. Such was the optimisation
generated by the discovery of the methods such as Shor's factoring
algorithm that many felt it would be a matter of time before
efficient quantum algorithms would be found for many natural
problems. This optimism has been largely diminished by the
realisation that efficient quantum algorithms exploit subtle -- and
still largely mysterious -- properties of delicate quantum
superposition. Despite these complications several flavours of
quantum algorithm have been developed exploiting, variously,
algebraic structures, symmetries, and geometry (see, eg.,
\cite{childs:2008a} for a recent review). Arguably some of these
algorithms have less practical utility than, say, Shor's factoring
algorithm, and constitute more of a proof of principle.

Recently an efficient quantum algorithm to solve systems of
\emph{linear} equations \cite{harrow:2008b} was described. This
promises to allow the solution of, eg., vast engineering problems.
This result is inspirational in many ways and suggests that quantum
computers may be good at solving more than linear equations. In this
paper we investigate this hope in the context of nonlinear ODEs and
we present an efficient quantum algorithm to integrate large sparse
systems of nonlinear ODEs.

\section{Preliminaries}
In this paper we are concerned with the solutions of a set of $n$
first-order nonlinear ODEs whose nonlinear terms are given by $n$
polynomials $f_\alpha(\mathbf{z})$ in $n$ variables $z_j$, $j=1, 2,
\ldots, n$, over $\mathbb{C}$. That is, we look for solutions
$\mathbf{z}(t)$ of the simultaneous set
\begin{equation}\label{eq:odes}
\begin{split}
\frac{dz_1(t)}{dt} &= f_1(z_1(t), z_2(t), \ldots, z_n(t)) \\
\frac{dz_2(t)}{dt} &= f_2(z_1(t), z_2(t), \ldots, z_n(t)) \\
&\vdots \\
\frac{dz_n(t)}{dt} &= f_n(z_1(t), z_2(t), \ldots, z_n(t)), \\
\end{split}
\end{equation}
subject to the boundary condition $\mathbf{z}(0) = \mathbf{b}$.
Standard results ensure that a solution to this initial value
problem exists and is unique (see, eg., \cite{arnold:1992a}).

We'll mostly only describe the algorithm for \emph{quadratic}
systems (the extension to higher degrees is straightforward).
Quadratically nonlinear equations can exhibit a wide variety of
phenomena, including, chaos and anomalous diffusion and classical
examples include the eponymous \emph{Lorenz system} and the
\emph{Orszag-McLaughlin dynamical system}.

We encode the variables $z_j(t)$ as the \emph{probability
amplitudes} of a quantum state of an $(n+1)$-level quantum system:
\begin{equation}
|\phi\rangle = \frac{1}{\sqrt{2}}| 0 \rangle +
\frac{1}{\sqrt{2}}\sum_{j=1}^n z_j|j\rangle,
\end{equation}
where, to ensure that the state is normalised, we require
$\sum_{j=1}^n |z_j|^2 = 1$. (Our constructions ensure that
$|\phi\rangle$ always remains normalised.) While it is convenient to
regard $|\phi\rangle$ as the state of a single $(n+1)$-level quantum
system, when actually implementing the algorithm on a quantum
computer we'll encode $|\phi\rangle$ as a state of $\log(n)$ qubits
in the natural way.

We describe our procedure in two stages: (1) a method to effect
nonlinear transformations of the probability amplitudes of
$|\phi\rangle$; (2) a quantum algorithm to implement Euler's method.

\section{A quantum algorithm to effect a nonlinear transformation of the
amplitudes}\label{sec:nonlintx} In this section we describe a
nondeterministic algorithm to prepare quantum states whose
amplitudes are nonlinear functions of those of some input quantum
state.

Suppose we want to effect a \emph{quadratic} transformation on the
probability amplitudes $z_j$ of $|\phi\rangle$. (We assume
throughout that the initial state $|\phi\rangle$ can be efficiently
prepared on a quantum computer, eg., using a combination of the
methods of \cite{rudolph:2002a} and \cite{harrow:2008b}.) Since the
amplitudes are unknown this is, in general, impossible and,
unfortunately, this is the generic setting when integrating ODEs as
every point on the solution trajectory can be regarded as an initial
condition for the system. But suppose we have two copies of
$|\phi\rangle$. In this case the probability amplitudes of the
tensor product are given by
\begin{equation}
|\phi\rangle|\phi\rangle = \frac{1}{2}\sum_{j,k = 0}^n z_jz_k
|jk\rangle,
\end{equation}
where, for convenience, we set $z_0 = 1$ from now on. Evidently
every monomial $z_j^{l_j}z_k^{l_k}$, $l_j, l_k \le 1$, of degree
less than $2$ appears (more than once) in this expansion. Suppose we
want to iterate the transformation
\begin{equation}
\mathbf{z} \mapsto F(\mathbf{z}),
\end{equation}
where
\begin{equation}
F(\mathbf{z}) = \begin{pmatrix} f_1(\mathbf{z}) \\ f_2(\mathbf{z})
\\ \vdots \\ f_n(\mathbf{z})
\end{pmatrix},
\end{equation}
and $f_\alpha$, $\alpha = 1, 2, \ldots, n$, are quadratic
polynomials
\begin{equation}
f_\alpha(\mathbf{z}) = \sum_{k,l=0}^n a^{(\alpha)}_{kl} z_kz_l,
\end{equation}
with $a^{(\alpha)}_{kl} = a^{(\alpha)}_{lk}$ and $f_0(\mathbf{z}) =
1$. Thus we aim to prepare the quantum state
\begin{equation}
|\phi'\rangle = \frac{1}{\sqrt{2}}\sum_{\alpha=0}^n
f_\alpha(\mathbf{z})|\alpha\rangle.
\end{equation}
where, to simplify matters, we've assumed that the transformation is
\emph{measure preserving}, i.e.,
\begin{equation}
1 = \sum_{j=1}^n |z_j|^2 = \sum_{\alpha=1}^n
f_\alpha^*(\mathbf{z})f_\alpha(\mathbf{z}).
\end{equation}
The measure-preservation assumption plays an important role in our
constructions, however, if it is relaxed our algorithm still
proceeds unchanged: only the success probability is modified.

To ensure that our method is efficient we need to make several extra
assumptions beyond measure preservation. The first assumption is
that $|a_{kl}^{(\alpha)}| = O(1)$, $k, l, \alpha = 1, 2, \ldots, n$.
The second assumption is that the map $F$ is \emph{sparse}, which
means that
\begin{equation}
|\{(k,l)\,|\, a_{kl}^{(\alpha)} \not=0\}| \le s/2, \quad \alpha = 1,
2, \ldots, n,
\end{equation}
\emph{and}
\begin{equation}
|\{\alpha\,|\, a_{kl}^{(\alpha)} \not=0\}| \le s/2, \quad k,l = 1,
2, \ldots, n,
\end{equation}
where $s = O(1)$. Note that the assumption of sparsity means that
each $f_\alpha(\mathbf{z})$ can only involve at most $s/2$ monomials
and that each variable appears in at most $s/2$ polynomials
$f_\alpha(\mathbf{z})$. The final assumption is that the Lipschitz
constant for our system, i.e., that number $\lambda$ such that
\begin{equation}
\|F(\mathbf{x}-\mathbf{y})\| \le \lambda \|\mathbf{x}-\mathbf{y}\|
\end{equation}
in the ball $\|\mathbf{x}\|^2\le 1$ and $\|\mathbf{y}\|^2\le 1$, is
$O(1)$. While these assumptions are rather restrictive, as we
discuss, there are still many important systems which satisfy them.
We also later describe how to relax these assumptions.

It turns out that implementing the desired transformation will, in
general, require that we make some destructive nonunitary
transformation of the system's state. To understand what is required
we now set up the operator
\begin{equation}
A =  \sum_{\alpha,k,l=0}^n a^{(\alpha)}_{kl} |\alpha0\rangle\langle
kl|.
\end{equation}
We now adjoin a qubit ``pointer'' $P$ and use $A$ to set up a
hamiltonian (this is essentially the von Neumann measurement
prescription \cite{peres:1993a}):
\begin{equation}
H = -iA\otimes |1\rangle_P \langle 0| + iA^\dag\otimes
|0\rangle_P\langle 1|.
\end{equation}
We now initialise our system in the state
\begin{equation}
|\phi\rangle|\phi\rangle|0\rangle_P,
\end{equation}
and evolve according to $H$ for a time $t=\epsilon$. The time we can
evolve for depends crucially on the sparsity of $H$ and the desired
error \cite{berry:2007a}; roughly speaking, when the sparsity of $A$
is some constant $s$ we can efficiently simulate (in terms of
$\log(n)$), up to some prespecified precision, the evolution for any
constant time $t$.

After the evolution the system ends up in the state
\begin{equation}\label{eq:expseries}
\begin{split}
|\Psi\rangle &= e^{i\epsilon H}|\phi\rangle|\phi\rangle|0\rangle =
\sum_{j=0}^\infty \frac{(i\epsilon H)^{j}}{j!} |\phi\rangle|\phi\rangle|0\rangle \\
&= |\phi\rangle|\phi\rangle|0\rangle + \epsilon
A|\phi\rangle|\phi\rangle |1\rangle - \cdots.
\end{split}
\end{equation}
Now noting that
\begin{equation}
A|\phi\rangle|\phi\rangle = \frac12\sum_{\alpha,k,l=0}^n
a_{kl}^{(\alpha)}z_kz_l|\alpha\rangle |0\rangle =
\frac{1}{\sqrt{2}}|\phi'\rangle |0\rangle
\end{equation}
we now measure $|\Psi\rangle$ on the ancilla qubit and postselect on
``$1$''; we succeed with probability $\approx\frac12\epsilon^2$
(thanks to the measure preservation property; if our polynomial map
is not measure preserving then this success probability will be
proportionally lower) and our posterior state is
\begin{equation}
\frac{\sqrt{2}}{\epsilon}(\mathbb{I}\otimes\mathbb{I}\otimes \langle
1|)|\Psi\rangle= \frac1{\sqrt2}\sum_{\alpha,k,l=0}^n
a_{kl}^{(\alpha)}z_kz_l |\alpha\rangle |0\rangle =
|\phi'\rangle|0\rangle.
\end{equation}
If we fail then we end up with some ``poisoned'' state
$|\Psi'\rangle$ (this occurs with probability $\approx
1-\frac12\epsilon^2$), which we discard. In order to ensure, with
high probability, that we end up with at least one copy of
$|\phi'\rangle$ we need to repeat this process on roughly
$16/\epsilon^2$ fresh pairs $|\phi\rangle|\phi\rangle$. It is an
interesting question whether one can design $H$ so that the poisoned
state can be recovered and used again. Such a possibility would
allow one to substantially reduce the resource requirements of our
algorithm.

(Obviously, because the expansion (\ref{eq:expseries}) is truncated
to first order, we don't exactly produce $|\phi'\rangle$, but rather
some approximation to $|\phi'\rangle$. To correct this we actually
use the method described in \cite{harrow:2008b} to implement the
transformation
\begin{equation}\label{eq:actualmap}
|\phi\rangle|\phi\rangle|0\rangle \mapsto
\sqrt{\mathbb{I}-\epsilon^2 H^2} |\phi\rangle|\phi\rangle|0\rangle +
i\epsilon H|\phi\rangle|\phi\rangle|0\rangle,
\end{equation}
where $|\epsilon| \le 1/\|H\|$. Since $H$ is sparse we have, by
Ger\v{s}gorin's theorem \cite{horn:1990a}, that $\|H\| \le cs$,
where $c$ is a constant. This allows us to assume that when we
succeed we will obtain precisely a copy of $|\phi'\rangle$, modulo
only imperfections in the simulation of $e^{i\epsilon H}$. We
discuss these errors in the appendix.)

If we want to iterate the polynomial map $\mathbf{z} \mapsto
F(\mathbf{z})$ a constant (in $n$) number $m$ times to produce the
state $|\phi^{(m)}\rangle$ we need to start with
$(\frac{16}{\epsilon^2})^m$ initial states. Thus the total spatial
resources required by this algorithm scale as
$(\frac{16}{\epsilon^2})^m\log(n)$ because the cost of storing $n$
variables via encoding in $|\phi\rangle$ scales linearly with
$\log(n)$. (See Proposition~\ref{pacc} for a precise statement of
the spatial resource requirements of our algorithm.)

The simulation of the evolution $e^{itH}$ on a quantum computer
cannot be done perfectly; we quantify, in
proposition~\ref{prop:errors}, the errors that accumulate throughout
the running of the polynomial iteration algorithm: if the evolution
$e^{itH}$ can be simulated up to error $\delta$, then after $m$
steps the final state will have accumulated an error no worse that
$\delta (3\gamma)^{m+1}$. Thus, choosing the simulation error to
satisfy $\delta < (3\gamma)^{-m}$ will ensure that the $m$th iterate
is exponentially close the desired state. The costs
\cite{berry:2007a} of simulation to this level of precision imply
that the total running time $T$ of our algorithm scales as $T\sim
m\poly(\log(n)\log^*(n)) s^2 9^{\kappa \sqrt{m}}$, where $\kappa$ is
some $O(1)$ constant and
\begin{equation}
\log^*(n) \equiv \min\{r\,|\, \log_2^{(r)}(n) < 2\}
\end{equation}
with $\log^{(r)}$ denoting the $r$th iterated logarithm, and we've
assumed that the transformation (\ref{eq:actualmap}) is carried out
in parallel on each of the remaining pairs in each iteration. Thus
the parallelised temporal scaling is subexponential in $m$ and
polynomial in $\log(n)$.

To complete the description of our iteration algorithm we need to
describe how to read out information about the solution. After $m$
iterations the system will be, up to some prespecified error, in the
quantum state $|\phi^{(m)}\rangle$ encoding the $m$th iterate of $F$
in the probability amplitudes. To access information about the
solution we need to make measurements of the system. In principle,
any hermitian observable $M = \sum_{j,k=0}^n M_{j,k}|j\rangle
\langle k|$ may be measured to extract information. Via Hoeffding's
inequality we learn that we can estimate, using repeated
measurements, the quantity
\begin{equation}
\langle M \rangle \equiv \sum_{j,k=0}^n \overline{z}_j M_{j,k} z_k
\end{equation}
to within any desired additive error. In practice the observable $M$
is measured using (a discretisation of) von Neumann's measurement
prescription \cite{childs:2002b, peres:1993a}, so the evolution
$e^{itM}$ must be efficiently simulable on a quantum computer. Many
natural such operators fall into this class, including -- via a
quantum fourier transform -- the operator whose measurement
statistics provide information on the sums
\begin{equation}
S_k = \frac{1}{\sqrt{n}}\sum_{j=1}^n x_j e^{\frac{2\pi i j k}{n}}.
\end{equation}
As noted in \cite{harrow:2008b}, and evident from our construction
here, by measurement of multiple copies of $|\phi^{(m)}\rangle$ it
is also possible to extract information about polynomial functions
of the solution.

The flexibility to measure efficiently implementable hermitian
operators to extract information about the solution provides the key
to the exponential separation between our method and the best
classical method. Indeed, if we only wanted to learn one element
$z_j$, for some $j$, of the $m$th iterate then there is actually an
efficient classical algorithm with the same resource scaling. (This
is similar to the situation with linear equations \cite{demko:1984a,
benzi:1999a} where if we only want to learn about one element of the
solution vector of a well-conditioned sparse set of linear equations
we can do this efficiently classically. Indeed, even if the system
is badly conditioned, we can still learn about parts of the solution
in the well-conditioned subspace.)

\section{Solving nonlinear differential equations}
In this section we show how to use the method we've just described
to integrate a sparse set of simultaneous nonlinear differential
equations for any constant time.

Suppose we want to integrate the system\footnote{We again assume
that our system is measure preserving which means that $\sum_{j=0}^n
z_j^*f_j(\mathbf{z}(t)) + z_jf_j^*(\mathbf{z}(t)) = 0$ and that the
Lipschitz constant is O(1).} (\ref{eq:odes}) with the initial
condition $\mathbf{z}(0) = \mathbf{b}$, where $f_j$ are sparse
polynomials. The simplest approach is to use \emph{Euler's method}
\cite{press:2007a}: we pick some small step size $h$ and iterate the
map
\begin{equation}
z_j \mapsto z_j + h z_j' = z_j + h f_j(\mathbf{z}).
\end{equation}
While Euler's method is pretty terrible in practice, especially for
stiff ODEs, it does provide a basic proof of principle; more
sophisticated methods such as 4th order Runge-Kutta and
predictor-corrector methods are essentially only \emph{polynomially}
more efficient. As our algorithm is precisely an implementation of
Euler's method on the probability amplitudes, we can appeal to the
standard theory (see, eg., \cite{iserles:1996a}) concerning its
correctness and complexity; thus our algorithm suffers from all the
standard drawbacks of Euler's method, including a 1st-order decrease
in error in terms of the step size $h$. Nonetheless, it will provide
us with an algorithm scaling polynomially with $\log(n)$, where $n$
is the number of variables. However, it should be noted that our
method scales \emph{exponentially} with the inverse step size. Thus,
without modification, our algorithm is really only suited to
well-conditioned systems.

As we've indicated, the idea behind our approach is very simple: we
integrate the system using Euler's method. Thus, given
$|\phi(t)\rangle$ we aim to prepare
\begin{equation}
|\phi(t+h)\rangle = |\phi(t)\rangle + h |\phi'(t)\rangle + O(h^2),
\end{equation}
where now
\begin{equation}
|\phi'(t)\rangle = \frac1{\sqrt2}\sum_{\alpha=0}^n
f_\alpha(\mathbf{z}(t))|j\rangle =
\frac1{\sqrt2}\sum_{\alpha,k,l=0}^n a_{kl}^{(\alpha)}z_k(t)z_l(t)
|\alpha\rangle.
\end{equation}
To implement this transformation we suppose we have two copies of
$|\phi(t)\rangle$ and apply the method of the previous section to
implement the polynomial transformation $z_\alpha \mapsto z_\alpha +
h f_\alpha(\mathbf{z}(t))$. Note that this transformation is only
measure preserving to $O(h)$; the success probability will be
diminished by a factor of $O(h^2)$, which can be made negligible by
reducing $h$ in the standard way.

So, to integrate the system (\ref{eq:odes}) forward in time to
$t=O(1)$ we begin by discretising time into $m$ steps. (Thus our
step size is $h=t/m$.) We then prepare $(\frac{16}{\epsilon^2})^m$
copies of $|\phi(0)\rangle$, where $\epsilon$ is as in
\S\ref{sec:nonlintx}, and apply the method of \S\ref{sec:nonlintx}
to produce approximately $(\frac{16}{\epsilon^2})^m$ copies of
$|\phi(t/m)\rangle$ in expected time $\mbox{poly}(\log(n))$. We then
iterate until we produce at least one copy of $|\phi(t)\rangle$ in
expected time $\mbox{poly}(m, \log(n))$ with probability greater
than $1/3$. The resources required by this approach scale
polynomially with $\log(n)$ and \emph{exponentially} with $t$ and
$1/h$.

\section{Extensions and applications}
In this section we briefly describe several extensions and
applications of our quantum Euler's method.

What sort of systems will be tractable with our approach? We only
sketch a couple of examples here, leaving the wider application of
our approach to more detailed investigations. Because the sparsity
and the measure-preserving properties play a key role in the
resource scaling of our algorithm it is desirable to focus on those
sparse systems preserving the ``hamiltonian'' $\sum_{j=1}^n
|z_j|^2$. One example of such a system is the
\emph{Orszag-McLaughlin dynamical system} \cite{orszag:1977a,
orszag:1980a}:
\begin{equation}
\frac{dx_j}{dt} = x_{j+1}x_{j+2} + x_{j-1}x_{j-2} - 2x_{j+1}x_{j-1},
\quad j = 1, \ldots, n,
\end{equation}
with periodic boundary conditions $x_{n+1} \equiv x_1$. The
variables $x_j$ are real and preserve $\sum_{j=1}^n x_j^2$. The
dynamics generated by this system are extremely complicated.

Another example of a system which can be studied using our algorithm
is the (discrete) \emph{nonlinear Schr{\"o}dinger equation} on any
finite graph $G = (V, E)$ of bounded degree:
\begin{equation}
-i\frac{dz_v}{dt} = 2\mbox{deg}(v)z_v - \sum_{w\sim v} z_w +
|z_v|^k z_v, \quad v \in V,
\end{equation}
where $k \in \mathbb{N}$.

Several extensions of our algorithm are possible. The first obvious
extension is to systems whose nonlinearity is cubic or higher. This
can be done in the natural way by consuming 3 (or more, for higher
degrees of nonlinearity) copies of $|\phi(t)\rangle$ at each step to
implement the desired nonlinear transformation. A second extension
is to certain densely defined systems, i.e., those for which the $A$
operator is dense: because, following \cite{harrow:2008b}, we can
actually implement any efficiently computable function $g$ of $A$ in
the step, we can access some dense operators $g(A)$. A third
extension allows the efficient computation of the equal-time
statistics of deterministic dynamical systems. Here the idea is to
apply our method not to $|\phi\rangle = \sum_{\alpha}
z_\alpha(t)|\alpha\rangle$, but rather to one half of an initial
state which is an entangled pair:
\begin{equation}
|\Phi\rangle = \int d\mu(\mathbf{z})
|\phi(\mathbf{z})\rangle|\mathbf{z}\rangle_P,
\end{equation}
where $\langle \mathbf{z}|\mathbf{z}'\rangle = \delta(\mathbf{z}-
\mathbf{z}')$ and $d\mu(\mathbf{z})$ is an efficiently implementable
probability measure, eg., uniform or gaussian \cite{rudolph:2002a}.
Applying our algorithm to this initial state allows us to
efficiently sample equal-time statistics via measurements on $P$.
This should be contrasted with the classical Hopf functional
approach \cite{hopf:1952a} to solving this problem which introduces
a Fokker-Planck type partial differential equation to study
distributions of solution trajectories. (The application of the Hopf
functional approach to the Orszag-McLaughlin system is considered in
\cite{ma:2005a}.)

Finally, it is not implausible that there is a trade-off between
time and space; perhaps there is a quantum algorithm which
integrates a constant number of variables which scales polynomially
with $\log(t)$ and $-\log(h)$?

\section{Conclusions and future directions}
We have presented a quantum algorithm to iterate large systems of
sparse polynomial maps. We've also described an implementation of
Euler's method to integrate a system of ODEs. As long as the system
is sparse the resources required by the method are polynomial in
$\log(n)$, where $n$ is the number of variables. However, the
resources consumed by the method scale exponentially with the
inverse step size and the integration time, as well as with the
degree of the nonlinearity.

\providecommand{\bysame}{\leavevmode\hbox
to3em{\hrulefill}\thinspace}
\providecommand{\MR}{\relax\ifhmode\unskip\space\fi MR }
\providecommand{\MRhref}[2]{%
  \href{http://www.ams.org/mathscinet-getitem?mr=#1}{#2}
} \providecommand{\href}[2]{#2}

\appendix
\section{Proofs of the claims}
In this appendix we present the technical proofs of the claims made
in the text.

\subsection{Description of the algorithm}
In this subsection we provide the formal specification of our
iteration algorithm.

\begin{algorithm}
\caption{}\label{alg:iteration}
\begin{algorithmic}
\STATE Set $N = (\frac{p}{\gamma})^{-m}$, according
to Proposition~\ref{pacc}; %
\STATE Initialise the system in the state $\ket{\Phi} :=
(\ket{\phi}\ket{\phi}\ket{0})^{\otimes N/2}$, where $\ket{\phi} =
\frac{1}{\sqrt{2}}|0\rangle + \frac{1}{\sqrt{2}}\sum_{j=1}^n
z_j|j\rangle$; %
\STATE Set $H :=-iA\otimes |1\rangle \langle 0| +
iA^\dag\otimes|0\rangle\langle 1|$; \FOR{$i = m$ to $1$ in steps of
$-1$} %
\STATE $S=0$; %
\FOR {$j=1$ to $N$ in steps of $1$} %
\STATE Evolve the $j$th pair $|\phi_{2j-1}\rangle|\phi_{2j}\rangle|0\rangle$ according to (\ref{eq:actualmap}), measure the ancilla, and postselect on ``1''. %
\STATE $\ket{\Phi'_j} :=  -i\sqrt{2}(\mathbb{I} \otimes \mathbb{I}
\otimes \langle 1|)H|\phi_{2j-1}\rangle|\phi_{2j}\rangle|0\rangle$; %
\STATE {\bf if} \emph{success} (i.e. $\ket{\Phi'_j} = | \phi'_{2j-1}
\rangle |0\rangle$) {\bf then} $S := S+1$; %
\ENDFOR %
\STATE $N:=2 \lfloor S/2 \rfloor$; %
\STATE {\bf if} $S<2^{i-1}$ {\bf then Exception:} algorithm failed; %
\STATE Set $\ket{\Phi} := (\ket{\phi'}\ket{\phi'}\ket{0})^{\otimes N/2}$;%
\ENDFOR
\end{algorithmic}
\end{algorithm}
\subsection{Bounding the expected running time}
We now prove the following
\begin{proposition}\label{pacc}
Let $m\ge 6$. If the number $N$ of initial states for
Algorithm~\ref{alg:iteration} with $p=\epsilon^2/2$ is
$(\frac{p}{16})^{-m}$ then it succeeds in producing at least one
state $|\phi^{(m)}\rangle$ with probability at least $1/3$.
\end{proposition}
\begin{proof}
Let $S$ be the random variable which counts the number of successes
out of the $N/2$ trials in each round and let $N_j$ denote the
number of successfully produced states in round $j$. Hoeffding's
inequality provides a bound on the cumulative distribution function
of a binomial random variable $S$ with parameters $(N,p)$:
$$F(k;N,p) = \mathbb{P}(S \leq k) \leq e^{-\frac{2(Np-k)^2}{N}}$$
for $0 \leq k \leq \mathbb{E} \left[ S \right] =
\frac{\epsilon^2}{2}N$. We declare failure if $S \leq k$, so that,
defining $k=\lambda N$, we have
\begin{eqnarray*}
\mathbb{P}(\mbox{failure}) &\leq&  e^{-\frac{2(Np-\lambda N)^2}{N}}\\
&=& e^{-2N(p-\lambda)^2}.
\end{eqnarray*}
Therefore
$$\mathbb{P}(\mbox{success}) \geq 1 - e^{-2N(p-\lambda)^2}.$$
We want the $m$-step algorithm to produce at least one state
$|\phi^{(m)}\rangle$ encoding the $m$th iterate, i.e.\ we want $N_m
\geq 1$ with probability at least $1/3$. That is, we require
\begin{eqnarray*}
\mathbb{P}(\mbox{$m$ successes}) &=& \mathbb{P}(\mbox{Success}_1)\mathbb{P}(\mbox{Success}_2) \cdots \mathbb{P}(\mbox{Success}_m)\\
&\geq& (1 - e^{-2N_1(p-\lambda)^2})(1 - e^{-2N_2(p-\lambda)^2})\cdots (1 - e^{-2N_m(p-\lambda)^2})\\
&\geq& 1/3.
\end{eqnarray*}
If the algorithm succeeds, after one step we produce at least
$\frac{k}{2} = \frac{\lambda N}{2}$ states ($k$ successes) with
probability at least $(1 - e^{-2N_1(p-\lambda)^2})$, and so after
$j$ steps we have that
\begin{equation} \label{n_jbound}
N_j \geq \left(\frac{\lambda}{2}\right)^{j-1} N
\end{equation}
with probability at least $(1 - e^{-2N_1(p-\lambda)^2})(1 -
e^{-2N_2(p-\lambda)^2})\cdots (1 - e^{-2N_j(p-\lambda)^2})$.

To ensure that the final success probability is greater than $1/3$
we demand that
$$\mathbb{P}(\mbox{Success}_j) \geq 1- e^{-2N_j(p-\lambda)^2} \geq 1 - \frac{1}{m}$$
so that
$$\mathbb{P}(\mbox{$m$ successes}) \geq \left(1 - \frac{1}{m}\right)^m = e^{-1}+ O(1/m)\geq 1/3,$$
for $m\ge 6$. Therefore, for all $j$ we need that
\begin{equation*}
1 - e^{-2N_j (p-\lambda)^2} \geq 1 - \frac{1}{m}
\end{equation*}
that is,
\begin{equation} \label{1/mbound}
\frac{1}{m} \geq e^{-2N_j(p-\lambda)^2}.
\end{equation}
The RHS of (\ref{1/mbound}) is decreasing in $N_j$ and so is maximum
when $j=m$. From (\ref{n_jbound}) we have
$$N_m \geq \left(\frac{\lambda}{2}\right)^{m-1}N$$
substitution into (\ref{1/mbound}) gives
$$\frac{1}{m} \geq e^{-2\left(\frac{\lambda}{2}\right)^{m-1}N (p-\lambda)^2},$$
that is, we require that
$$\log m \leq 2\left(\frac{\lambda}{2}\right)^{m-1} N (p-\lambda)^2.$$
If we choose $\lambda = p/2$ and $N=(\frac{p}{8})^{-m}
=(\frac{\lambda}{4})^{-m}$ then from (\ref{n_jbound}) we have
\begin{eqnarray*}
N_m &\geq& \left(\frac{\lambda}{2}\right)^{m-1}N\\
&=& \left(\frac{\lambda}{2}\right)^{m-1} \left(\frac{\lambda}{4}\right)^{-m}\\
&=& 2^m \frac{2}{\lambda},
\end{eqnarray*}
and the result follows.
\end{proof}

\subsection{Bounding the accumulated error}
We now bound the accumulated error that builds up throughout the
running of the quantum iteration algorithm.

\begin{proposition}\label{prop:errors}
The error $\delta_m = \||\phi^{(m)}\rangle - |\psi^{(m)}\rangle\|$
that accumulates after $m$ iterations of the quantum iteration
algorithm~\ref{alg:iteration} is bounded by
$$\delta_m \le \frac{\eta}{3} \left(\frac{(3\gamma)^{m+1}-1}{3\gamma -
1} - 1\right)$$ where $\eta$ is the error in simulating
$e^{i\epsilon H}$, $|\psi^{(m)}\rangle$ is the actual state produced
by the algorithm, and $\gamma$ is an $O(1)$ constant which depends
on the sparsity $s$.
\end{proposition}

\begin{proof}
Our proof works by analysing the errors that accumulate in pairs
$|\phi_j\rangle|\phi_j\rangle$ of states during the $j$th round.
Suppose we have an error in the starting pair, so instead of
$\ket{\phi_0}$ we have $\ket{\phi_0} + \ket{\Delta \phi_0} =
\ket{\psi_0}$, say, where the initial error $\ket{\Delta\phi_0}$ has
magnitude $\delta_0$, i.e.\ $\delta_0 = \| \ket{\Delta\phi_0} \|$.
Suppose also that our simulation of $U = e^{i\epsilon H}$ is
imperfect, i.e.\ $V$ is the operator that is actually applied and $
\| U-V \|_\infty \leq \eta$.

We initialise our system into the state
$\ket{\psi_0}\ket{\psi_0}\ket{0}$ and evolve according to $V$. This
particular pair will then be in the state
\begin{eqnarray*}
V \ket{\psi_0} \ket{\psi_0} \ket{0}
\end{eqnarray*}
To measure this state on the ancilla qubit and postselect on ``$1$"
we apply the measurement operator $P_1 = \mathbb{I} \otimes
\mathbb{I} \otimes \ket{1}\bra{1}$. The (subnormalised) posterior
state of the pair is then
\begin{eqnarray*}
\ket{\overline{\psi}_1} = P_1 V (\ket{\phi_0}+\ket{\Delta
\phi_0})(\ket{\phi_0}+\ket{\Delta\phi_0})\ket{0}.
\end{eqnarray*}
Recall that if there were no errors in the starting state and the algorithm was perfect we would have
\begin{eqnarray*}
\ket{\phi_1} = \frac{\sqrt{2}}{\epsilon}P_1 U
\ket{\phi_0}\ket{\phi_0}\ket{0}.
\end{eqnarray*}
Let the error between the subnormalised posterior states
$\ket{\overline{\psi}_1}$ and $\ket{\overline{\phi}_1}$ after one
step have magnitude $\overline{\delta}_1= \| \ket{\overline{\psi}_1}
- \ket{\overline{\phi}_1} \|$. We bound this as follows
\begin{equation}
\begin{split}
\overline{\delta}_1 &\leq \| P_1 V (\ket{\Delta\phi_0}\ket{\phi_0} +\ket{\phi_0}\ket{\Delta\phi_0} + \ket{\Delta\phi_0}\ket{\Delta\phi_0}) + P_1(V-U)  \ket{\phi_0}\ket{\phi_0}\|\\
&\leq \| \ket{\Delta\phi_0}\ket{\phi_0} +\ket{\phi_0}\ket{\Delta\phi_0} + \ket{\Delta\phi_0}\ket{\Delta\phi_0}\| +  \|(V-U)\|_\infty\\
&\leq \delta_0 + \delta_0 + \delta_0^2 + \eta \le 3\delta_0 + \eta.
\end{split}
\end{equation}
We bound the error $\delta_1$ between the normalised posterior
states by
\begin{equation}
\begin{split}
\delta_1 &= \||\psi_1\rangle -|\phi_1\rangle \| \le \frac{2}{\|
|\phi_1\rangle \|}\overline{\delta}_1 \\
&= \frac{2\sqrt{2}}{\epsilon}(3\delta_0 + \eta) = \gamma(3\delta_0 +
\eta),
\end{split}
\end{equation}
where $\gamma$ is, by assumption, a constant of order $s$.

Repeating this argument allows us to set up the recurrence
\begin{equation}
\delta_j \le \gamma (3 \delta_{j-1}+\eta), \quad j=1, 2, \ldots, m.
\end{equation}
Solving the recurrence, and assuming that the initial states are
constructed perfectly (i.e., $\delta_0 = 0$) gives us
\begin{equation}
\delta_m \le \frac{\eta}{3} \left(\frac{(3\gamma)^{m+1}-1}{3\gamma -
1} - 1\right).
\end{equation}

\end{proof}

\end{document}